\documentclass[twoside]{src/main/latex/main/LNGAI}

\usepackage{src/main/latex/main/LNGAImacro}

\usepackage{graphicx}
\graphicspath{{images/}}

\usepackage{amsmath}
\usepackage{amssymb}

\def\lastname{Ganguly, Mendez and Kampik}

\usepackage[utf8]{inputenc}
\usepackage{soul}\setuldepth{article}
\usepackage{centernot}
\usepackage{enumerate}
\usepackage{tikz}
\usepackage{subfig}
\usepackage{mathtools}
\usepackage{thmtools}
\usepackage{thm-restate}
\usepackage{orcidlink}
\usepackage[normalem]{ulem}

\usepackage{hyperref}

\usepackage{makecell}

\newcommand{\orcidIDlink}[1]{\href{https://orcid.org/#1}{\orcidlink{#1} \texttt{{https://orcid.org/#1}}}}

\definecolor{oceanblue}{rgb}{0, 0.4824, 0.6549}

\makeatletter
\renewcommand{\p@subfigure}{}
\makeatother

\def\Args{\ensuremath{\mathit{Args}}}
\def\Att{\ensuremath{\mathit{Att}}}
\def\Supp{\ensuremath{\mathit{Supp}}}
\def\is{\ensuremath{\tau}}
\def\fs{\ensuremath{\sigma}}

\def\path{\ensuremath{\mathcal{P}}}
\def\interval{\ensuremath{\mathbb{I}}}

\def\arga{\ensuremath{\mathsf{a}}}
\def\argb{\ensuremath{\mathsf{b}}}
\def\argc{\ensuremath{\mathsf{c}}}
\def\argd{\ensuremath{\mathsf{d}}}
\def\arge{\ensuremath{\mathsf{e}}}
\def\argf{\ensuremath{\mathsf{f}}}
\def\argx{\ensuremath{\mathsf{x}}}
\def\argy{\ensuremath{\mathsf{y}}}
\def\argz{\ensuremath{\mathsf{z}}}

\newcommand{\argnode}[3]{\mbox{\ensuremath{#1~(#2)\!:\!\mathbf{#3}}}}

\tikzset{
    noanode/.style={scale=0.55,dashed, circle, draw=black!60, minimum size=10mm, font=\bfseries},
    unanode/.style={scale=0.55,circle, draw=black!75, minimum size=10mm, font=\bfseries},
    xunanode/.style={scale=0.55,circle, draw=black!75, minimum size=10mm, font=\bfseries, line width=2.5pt},
    invnode/.style={scale=0.55,circle, draw=white!0, minimum size=0mm, font=\bfseries},
    anode/.style={scale=0.55,circle, draw=lightgray!60, minimum size=10mm, font=\bfseries},
    xanode/.style={scale=0.55, circle, fill=lightgray, draw, minimum size=10mm, font=\bfseries,line width=2.5pt},
    xnoanode/.style={scale=0.55, circle, dashed, draw, minimum size=10mm, font=\bfseries,line width=2.5pt},
    sxnoanode/.style={scale=0.55, circle, dotted, draw, minimum size=10mm, font=\bfseries, line width=0.85pt}
}

\usepackage{xifthen}
\usepackage{xargs}
\newcommandx{\scon}[4][3=, 4=]{  
    \relax
    \ifmmode
        \ifthenelse{\isempty {#3}}
        { {\ensuremath{#1\ \sim #2}} }
        { {\ensuremath{#1\ \sim_{\sigma,#3,#4} #2}} }
    \else
        \ifthenelse{\isempty {#3}}
        { {\ensuremath{#1\ \sim #2}} }
        { {\ensuremath{#1\ \sim_{\sigma,#3,#4} #2}} }
    \fi
}
\newcommandx{\nscon}[4][3=, 4=]{
    \relax
    \ifmmode
        \ifthenelse{\isempty {#3}}
        { {\ensuremath{#1\ \not\sim #2}} }
        { {\ensuremath{#1\ \not\sim_{\sigma,#3,#4} #2}} }
    \else
        \ifthenelse{\isempty {#3}}
        { {\ensuremath{#1\ \not\sim #2}} }
        { {\ensuremath{#1\ \not\sim_{\sigma,#3,#4} #2}} }
    \fi
}

\begin{document}

    \begin{frontmatter}              

        \title{Safety, Oscillation, Liveness, and Fairness in Quantitative Argumentation Dialogues}

        \author{Arunavo~Ganguly}\footnote{Arunavo~Ganguly~\orcidIDlink{0009-0004-9098-9904}, \href{mailto:aganguly@cs.umu.se}{\texttt{aganguly@cs.umu.se}}},
        \author{Julian~Alfredo~Mendez}\footnote{Julian~A.~Mendez~\orcidIDlink{0000-0002-7383-0529},  \href{mailto:julian.mendez@cs.umu.se}{\texttt{julian.mendez@cs.umu.se}}},
        \author{Timotheus~Kampik}\footnote{Timotheus~Kampik~\orcidIDlink{0000-0002-6458-2252}, \href{mailto:tkampik@cs.umu.se}{\texttt{tkampik@cs.umu.se}}}

        \address{Umeå University, Umeå, Sweden}

        \begin{abstract}
            We introduce the notions of \emph{safety}, \emph{liveness}, and \emph{fairness}, as commonly used in temporal reasoning and distributed systems, to quantitative (bipolar) argumentation dialogues where repeated inferences are drawn from argumentation graphs with weighted nodes.
            Between inferences, these graphs are updated.
            Safety and liveness capture whether arguments' final strengths always and eventually (respectively) attain a specific threshold of credibility.
            Fairness notions assess how the safety of arguments is distributed within a sequence of argumentation graphs.
            Additionally, we introduce the notion of \emph{oscillation} to capture the stability of topic arguments with respect to the threshold of credibility.
            We formally show how these notions are related and why providing general guarantees for these properties is analytically challenging.
        \end{abstract}

    \end{frontmatter}


    \section{Introduction}
    \label{sec:intro}
    Quantitative (bipolar) argumentation is a formal argumentation variant in which inferences are drawn from graphs whose nodes (\emph{arguments}) with weights (\emph{initial strengths}) are connected by \emph{support} and \emph{attack} relations.
    Argumentation \emph{semantics} then draw inferences from a \emph{quantitative bipolar argumentation graph} (QBAG) by updating arguments' initial strengths to \emph{final strengths}, considering  the QBAG's topology.
    QBAGs are of substantial interest to the formal argumentation community, notably because of their potential to complement machine learning approaches and help make them \emph{explainable}~\cite{Potyka0T23,DBLP:conf/ecai/0007PT23} and \emph{contestable}~\cite{yin2025contestabilityquantitativeargumentation}.
    QBAG applications are emerging that facilitate the explainability and reliability of large language model-based inference~\cite{pmlr-v284-zhu25a,DBLP:conf/aaai/FreedmanDG00T25,jin2026argoraorchestratedargumentationcausally}.
    In some applications, argumentation can be viewed as a dynamic process in which QBAGs are manipulated, and repeated inferences are drawn.
    Accordingly, it may be interesting to inquire about the dynamic properties of QBAG-based inference.

    This paper systematizes such inquiries by drawing inspiration from well-known notions of \emph{safety}, \emph{liveness}, and \emph{fairness}, which are prominent in various temporal reasoning approaches, such as linear temporal logic~\cite{Alpern-Schneider,DBLP:conf/popl/GabbayPSS80}, distributed processes~\cite{fairness_motivation}, and Petri nets~\cite{DBLP:journals/ipl/KindlerA99}\footnote{We do not formally relate to these notions of \emph{tracial control}, as we are interested in conceptually somewhat different questions, namely the dynamics of quantitative argument strengths.}.
    In addition, we introduce the notion of \emph{oscillation}, which aims to capture the stability of a topic argument.
    Collectively, we refer to these notions as \emph{SOLF}.

    Our notions of SOLF are introduced against the backdrop of a \emph{threshold of credibility}.
    Intuitively, arguments whose final strengths \emph{attain} the threshold are deemed credible, whereas arguments that fail to meet the threshold are not.
    Our SOLF notions aim to assess the credibility of a set of topic arguments over a dynamic process.
    This process is specified by a sequence of QBAGs (which we call a \emph{chain} or \emph{dialogue}). The \emph{topic arguments} are a subset of arguments occurring in all of the sequence's graphs.
    Note that we assume the threshold of credibility is use-case dependent, analogous to the interpretation of final strengths in QBAGs in general.

    Given this backdrop, our notions allow us to ask the following questions:
    \begin{description}
        \item[Safety:]
        Do all of the topic arguments always attain the threshold?
        This notion of safety can, for instance, help us determine whether obviously credible arguments always retain reasonable strengths in the dialogue, or whether the dialogue derails.
        This interpretation of safety aligns with the canonical Linear Temporal Logic (LTL) reading, namely that \emph{``bad things'' do not happen during execution},
        where ``bad things'' refers to the topic arguments falling below the threshold~\cite{Alpern-Schneider}.
        \item [Oscillation:]
        If the topic arguments are credible at the beginning, do they retain their credibility throughout the process?
        Oscillation describes the behavior of the topic arguments' final strengths throughout the dialogue, i.e., whether they remain stable or flip-flop across the threshold.
        Intuitively, this reflects that such arguments are subject to substantial debate.
        \item[Liveness:]
        Do all of the topic arguments eventually become credible, i.e., do they attain the threshold at the end of the dialogue?
        Again, this interpretation of liveness is aligned with the standard reading in LTL, namely \emph{``good things'' do happen $({\it eventually})$},
        where ``good things'' refers to the topic arguments attaining the threshold~\cite{Alpern-Schneider}.
        \item[Fairness:]
        Do all topic arguments get to have an equal chance to attain credibility?
        We present several discrete and gradual fairness notions, which are based on safety.
    \end{description}

    We informally demonstrate our intuitions in the following example.
    \begin{example}
        \label{ex:example-intro}
        \begin{figure}[ht]
            \centering
            \subfloat[Initial QBAG $G$.]{
                \centering
                \begin{tikzpicture}[node distance=1.5cm]
                    \node[unanode] (a) at (0, 0) {\argnode{\arga}{0.5}{0.6}};
                    \node[unanode] (b) at (2, 0) {\argnode{\argb}{0.7}{0.7}};
                    \node[unanode] (c) at (0, 2) {\argnode{\argc}{0.2}{0.2}};
                    \node[invnode] (i) at (3,-1) {\phantom{e}};
                    \node[invnode] (i) at (-0.5,-1) {\phantom{e}};

                    \draw[-stealth, thick] (c) -- node[pos=0.5, left=1pt] {+} (a);

                \end{tikzpicture}
                \label{fig:intro-initial}
            }
            \centering
            \subfloat[Normal expansion $G'$.]{
                \begin{tikzpicture}[node distance=1.5cm]
                    \node[unanode] (a) at (0, 0) {\argnode{\arga}{0.5}{0.1}};
                    \node[unanode] (b) at (2, 0) {\argnode{\argb}{0.7}{0.0}};
                    \node[unanode] (c) at (0, 2) {\argnode{\argc}{0.2}{0.2}};
                    \node[unanode] (d) at (2, 2) {\argnode{\argd}{1.0}{1.0}};
                    \node[invnode] (i) at (3,-1) {\phantom{e}};
                    \node[invnode] (i) at (-0.5,-1) {\phantom{e}};

                    \draw[-stealth, thick] (c) -- node[pos=0.5, left=1pt] {+} (a);
                    \draw[-stealth, thick] (d) -- node[pos=0.5, left=1pt] {-} (a);
                    \draw[-stealth, thick] (d) -- node[pos=0.5, left=1pt] {-} (b);
                \end{tikzpicture}
                \label{fig:intro-normal-expansion-1}
            }
            \centering
            \subfloat[Normal expansion $G''$.]{
                \begin{tikzpicture}[node distance=1.5cm]
                    \node[unanode] (a) at (0, 0) {\argnode{\arga}{0.5}{0.5}};
                    \node[unanode] (b) at (2, 0) {\argnode{\argb}{0.7}{0.56}};
                    \node[unanode] (c) at (0, 2) {\argnode{\argc}{0.2}{0.2}};
                    \node[unanode] (d) at (2, 2) {\argnode{\argd}{1.0}{0.2}};
                    \node[unanode] (e) at (1, 3.2) {\argnode{\arge}{0.8}{0.8}};
                    \node[invnode] (i) at (3,-1) {\phantom{e}};
                    \node[invnode] (i) at (-0.5,-1) {\phantom{e}};

                    \draw[-stealth, thick] (c) -- node[pos=0.5, left=1pt] {+} (a);
                    \draw[-stealth, thick] (d) -- node[pos=0.5, left=1pt] {-} (a);
                    \draw[-stealth, thick] (d) -- node[pos=0.5, left=1pt] {-} (b);
                    \draw[-stealth, thick] (e) --
                    node[pos=0.5, left=1pt] {-} (d);
                \end{tikzpicture}
                \label{fig:intro-normal-expansion-2}
            }
            \caption{From left to right, the initial QBAG expands by successively adding $\argd$ (normal expansion $G'$) and $\arge$ (normal expansion $G''$), respectively.
            Such scenarios capture the essence of \emph{dynamic argumentation}.
            }
            \label{fig:intro}
        \end{figure}
        Consider the QBAGs in Figure~\ref{fig:intro}.
        Each argument is labelled
        $\argx \ (\argy) : \argz$,
        denoting the name, initial strength, and final strength, respectively.
        Intuitively, initial strengths represent the credibility of arguments per se,
        whereas arguments' final strengths reflect the justification reached through argument interactions.
        We denote $\fs_{G}(\argx)$ as the final strength of $\argx$ in the QBAG $G$ and apply DFQuAD semantics, as defined in the preliminaries.
        Interactions are represented via attack and support relations, labelled $-$ and $+$, respectively.
        The QBAGs in Figure~\ref{fig:intro} portray a \emph{normal expansion chain}: new arguments and relationships are added, but relationships between existing arguments (and arguments' initial strengths) remain unchanged.

        Our SOLF notions are introduced with a \emph{threshold of credibility};
        credible arguments are assumed to have final strengths that attain this threshold.
        We understand safe argument sets to be credible throughout the process.
        For instance, $\{\argc\}$ is safe for the threshold $t=0.1$ in Figure~\ref{fig:intro}.
        Liveness stipulates that arguments always become credible eventually, i.e., at the end of the dialogue chain, their final strengths attain the threshold.
        $\{\argb\}$ is live w.r.t. the threshold of $0.1$, since $\fs_{G''}(\argb) \geq 0.1$.
        Extending this observation, $\{\argb\}$ also \emph{oscillates} w.r.t. $0.1$, since its final strength is below $0.1$ in $G'$, but attains it again in $G''$, thus crossing the threshold.
        Our notions of fairness then cover two dimensions: several binary classifications based on safety, and gradual notions based on the established measures of the Gini index and Shannon entropy.

        In the example, $\{\arga, \argb\}$ is \emph{cautiously fair}, which requires that if any of the topic arguments becomes credible (at the end), so does the whole set.
        Indeed, $\{\arga, \argb\}$ is cautiously fair because at least $\{\arga\}$ or $\{\argb\}$ is live.
        However, our gradual fairness notions provide more nuance: $\{\arga, \argb, \argc\}$ yields fairness scores of $\approx 0.98$
        and $\approx 0.46$ 
        for the \emph{Shannon-based} and \emph{Gini-based} notions (defined in Section~\ref{subsec:Fairness}), respectively.
        Intuitively, this occurs because $\{\argc\}$ has a better chance at justification than $\{\arga\}$ or $\{\argb\}$, indicating imperfect uncertainty (Shannon entropy) and inequality of justification distribution among arguments (Gini index).
    \end{example}
    \newpage
    In the following sections, we introduce relevant preliminaries (Section~\ref{sec:Prelims}), to then formally present our SOLF notions, observing some of their properties, as well as how the notions are interrelated (Section~\ref{sec:Safe-Live-Fair}).
    We also discuss our results in the context of related research (Section~\ref{sec:discussion-and-conclusion}) and sketch why giving guarantees w.r.t. the SOLF notions may be challenging using an empirical demonstration of semantics behavior, thus highlighting directions for future work (Section~\ref{sec:future-work}).

    \section{Preliminaries}
    \label{sec:Prelims}

    Let $\langle \mathbb{I},\preceq \rangle$ be a partial order.
    Typically, $\mathbb{I}$ will be a real interval alongside the canonical order (e.g., $\mathbb{I} = [0, 1]$).
    A \emph{quantitative bipolar argumentation graph} (QBAG) represents arguments with \emph{initial strengths} in $\mathbb{I}$, and support and attack relations between these arguments.
    Intuitively, an initial strength may represent
    the initial confidence in the argument without any influence; greater initial strengths represent higher confidence.
    \begin{definition} [QBAG]
        A \textit{Quantitative Bipolar Argumentation Graph} (QBAG) is a quadruple $\langle \Args, \tau, \Att, \Supp \rangle$, where $\Args$ is a non-empty finite set of elements called \emph{arguments}, $\Att, \Supp \subseteq \Args \times \Args$ are the \emph{attack} and the \emph{support} relations, respectively, over the arguments, with $\Att \cap \Supp=\varnothing$, and $\tau: \Args \rightarrow \mathbb{I}$ is a total function assigning an \emph{initial strength} $\tau(\argx)$ to each argument $\argx \in \Args$.
    \end{definition}
    Figure~\ref{fig:intro} features some QBAGs that are explained in Example~\ref{ex:example-intro}.
    Unless stated otherwise, let
    $G = \langle \Args, \tau, \Att, \Supp \rangle$ be an arbitrary but fixed QBAG, with finite $\Args$.
    Given an argument $\argx\in  \Args$, the attackers and supporters of $\argx$ are denoted by
    $$
    \Att_G(\argx):= \{\arga\in \Args: \langle \arga, \argx \rangle \in \Att\},$$
    $$ \Supp_G(\argx):= \{\arga\in \Args: \langle \arga,\argx \rangle \in \Supp\},
    $$
    respectively.
    We drop the subscript $G$ whenever it is clear from context.
    Given two arguments $\argx, \argy \in \Args$, we say that $\argx$ \emph{reaches} $\argy$ if and only if there is a path from $\argx$ to $\argy$ in the graph $\langle \Args, \Att \cup \Supp \rangle$.
    An \emph{acyclic} QBAG does not contain any loops, i.e., no argument $\argx\in \Args$ can reach itself.
    For any $\Args' \subseteq \Args$, we define the \emph{restriction of} $G$ \emph{to} $\Args'$ as \\
    $
    G\downarrow_{\Args'}:=\langle \Args', \tau \cap (\Args' \times \mathbb{I}), \Att \cap (\Args' \times \Args'), \Supp \cap (\Args' \times \Args') \rangle.
    $

    As we are interested in dynamic argumentation, we require notions that relate several QBAGs.
    \begin{definition} [Sub-QBAG]
        We say that $G'=\langle \Args', \tau', \Att', \Supp' \rangle$ is a \emph{sub-QBAG of} $G$, denoted by  $G'\sqsubseteq G$, if and only if $\Args'\subseteq \Args$, $\Att'\subseteq \Att$, $\Supp'\subseteq \Supp$, and $\tau'(\argx)=\tau(\argx)$, for every $\argx \in \Args'$.
        If $G' \sqsubseteq G$ but $G' \not = G$, we may denote this by $G' \sqsubset G$.
    \end{definition}

    Our idea of argumentation dialogues is captured via QBAG chains. Intuitively, chains represent evolutions of a given QBAG along arbitrary dimensions, such as expansions and initial strength updates.
    \begin{definition}[QBAG Chains]\label{def:chains}
        A \emph{QBAG chain} (short: \emph{chain}) is a sequence
        of QBAGs ${\cal G} = (G_1, \dots, G_n)$.
    \end{definition}
     Given a chain ${\cal G} = (G_1, \dots G_n)$,
    when we write $G \in \cal G$, we mean that for some $G_i$, $1 \leq i \leq n$, it holds that $G_i = G$.
    We may call $G_1$ and $G_n$ the \emph{start point} and \emph{end point} of $\cal G$, respectively.

    The notion of QBAG chains capture a broad range of QBAG modifications: attack and support additions/removals, strength change updates and expansions.
    If a chain only encompasses expansions, i.e., the addition of new arguments to the previous QBAG in the chain, we call it a \emph{QBAG expansion chain}.
    \begin{definition}[QBAG Expansion Chains]
        A \emph{QBAG expansion chain} (or simply \emph{expansion chain}) is a chain $(G_1, \dots, G_n)$, with $G_i \sqsubset G_{i+1}$ for $1 \leq i < n$.
        An expansion chain $\cal G$ is called:
        \begin{enumerate}
            \item [\emph{$(i)$}]
            \emph{normal}
            if and only if for every $G_i\! = \langle \Args, \is, \Att, \Supp \rangle$ and $G_j = \langle \Args', \is', \Att', \Supp' \rangle,$ with $G_i, G_j\in \cal G$ and $G_i\sqsubset G_j$ it holds that
            $\forall \langle \argx , \argy \rangle \in (\Att' \cup \Supp')\setminus( \Att \cup \Supp)$, either $\argx$ or $\argy$ is in $\Args' \setminus \Args$;
            \item [\emph{$(ii)$}]
            \emph{weak} if and only if $\cal G$ is normal and for every $G_i\! = \langle \Args, \is, \Att, \Supp \rangle$ and $G_j = \langle \Args', \is', \Att', \Supp' \rangle,$ with $G_i, G_j\in \cal G$ and $G_i\sqsubset G_j$ it holds that
            $\forall \langle \argx , \argy \rangle \in (\Att' \cup \Supp')\setminus( \Att \cup \Supp): \argy \notin \Args$.
        \end{enumerate}
    \end{definition}
    \begin{example}
        \label{ex:expansion-set-example}
        Referring to Figure~$\ref{fig:intro}$, ${\cal G} = (G, G', G'')$ is a (\emph{normal} but not \emph{weak}) expansion chain, where $G$ and $G''$ are the start and end points of $\cal G$, respectively.
    \end{example}
    Note that while the examples provided in this paper focus on expansion chains, our SOLF notions are applicable to arbitrary chains.

    \begin{definition}[Gradual semantics]
        A \emph{gradual semantics} is a function $\fs_G: \Args \rightarrow \mathbb{I} \cup \{\perp\}$ mapping the arguments in a QBAG $G$ to the partial order $\mathbb{I}$, along with the symbol $\perp$ for undefined~\cite{baroni-toni-rago}.
    \end{definition}
    \emph{Modular semantics} are well-studied sub-class of gradual semantics, where the final strength of an argument is assigned by \emph{aggregating} the (final)
    strengths of its attackers and supporters, and then computing the \emph{influence} on the argument given the aggregation result and the argument's initial strength.
    More formally, an \emph{aggregation function} has the form $f(A, S)$, given multi-sets of attack values $A$ and support values $S$.
    The application of an influence function $g_s(f(A, S))$ then yields an update of an argument's strength value $s \in \interval$.
    Modular semantics first \emph{initialize} all arguments by assigning their initial strengths as the temporary strength value, and then update these values by applying aggregation and influence function to these values following the topological order of the QBAG (in the direction of the binary relations).

    In our examples, we make use of DFQuAD semantics~\cite{Rago.al-2016}, given $\interval = [0, 1]$.
    \begin{definition}[DFQuAD]
        Let $A = \Att(\argx)$ and $S = \Supp(\argx)$, for $\argx \in \Args$.
        For the aggregation function
        \[
            f(A,S)=\prod_{\argb \in A}(1-\fs(\argb))-\prod_{\argc \in S}(1-\fs(\argc))\]
        and initial strength $\tau(\arga)$,
        DFQuAD is defined over acyclic graphs as per the equation given below:
        \begin{align*}
            g(\tau(\arga), f(A, S)) := \ & \tau(\arga)-\tau(\arga)\cdot \max\{0,-f(A, S)\} \ +
            \\
            & (1-\tau(\arga))\cdot \max\{0,f(A, S)\}.
        \end{align*}
    \end{definition}

    \noindent \textbf{Assumption.} Notice that for cyclic QBAGs, modular semantics may not converge and thus yield undefined final strengths in some cases~\cite{potyka2024empirical}.
    Because of this, we focus on acyclic QBAGs, which is a common assumption in the literature~\cite{kotonya2019gradual,cocarascu2019extracting,lidecision,chi2021optimized}, i.e., we assume that final strengths are comparable.


    \section{Safety, Oscillation, Liveness, and Fairness Properties}
    \label{sec:Safe-Live-Fair}

    In this section, we formally introduce our SOLF notions for QBAGs.
    The notions are validated w.r.t. an initial QBAG $G$, a chain of updates $\cal G$, a semantics $\sigma$, and a threshold $t\in \mathbb{I}$, which we call the \emph{threshold of credibility}.
    Throughout this section, we fix a partial order $\mathbb{I}$, arbitrary QBAGs $G=\langle \Args, \tau, \Att,\Supp \rangle$,
    $G'=\langle \Args', \tau', \Att',\Supp'\rangle$,
    $G''=\langle \Args'',\tau'',\Att'',\Supp''\rangle$, and a QBAG chain ${\cal G}=(G_1, \dots, G_n)$.
    We also fix a non-empty finite set of topic arguments $T\subseteq\bigcap_{1\leq i\leq n}\Args_{G_i}$, i.e., the arguments occurring in the entire chain, to analyze the SOLF notions.


    \subsubsection*{Safety and Liveness}
    \label{sec:safety}
    Intuitively, safe arguments remain credible throughout the whole span of the dialogue, i.e., their final strengths remain above the threshold of credibility when inferences are repeatedly drawn.
    \begin{definition}[Safety]
        \label{def:safety}
        We say that a topic set $T$ is safe w.r.t. $\cal G$, $\fs$, and $t$ if and only if for every $G\in \cal G$, and $\argx\in T$, it holds that $\fs_{G}(\argx)\geq t$.
    \end{definition}
    When we discuss whether a topic set is safe w.r.t. some chain $\cal G$, semantics $\fs$, and threshold $t$, we may leave the specification of any of $\cal G$, $\fs$, or $t$  implicit. For other SOLF notions, we proceed analogously; notably, in our examples, we generally assume DFQuAD semantics.

    When assessing liveness, we are interested only in the credibility of the topic arguments at the end of the dialogue, regardless of their credibility throughout the process.
    \begin{definition}[Liveness]
        \label{def:weak-safety}
        We say that a topic set $T$ is live w.r.t. $\cal G$, $\fs$, and $t$ if and only if  $T$ is safe w.r.t. the (singleton) QBAG chain $(G_n)$, $\fs$, and  $t$.
    \end{definition}

    We revisit our initial example to illustrate the safety and liveness notions.
    \begin{example}
        \label{ex:safety-example}
        In Figure~\ref{fig:intro}, the arguments $\{\argc\}$ and $\{\argb\}$ are safe and live, respectively, w.r.t. $(G, G', G'')$ and the threshold $t = 0.1$.
        However, if the threshold of credibility is changed to $t = 0.3$, $\{\argc\}$ is neither safe nor live, whereas $\{\argb\}$ continues to be live.
    \end{example}

    Since safety concerns the credibility of topic arguments over the whole process, it is immediate that safety implies liveness.
    \begin{theorem}
        \label{thm:strong-implies-weak-safety}
        If a topic set $T$ is safe w.r.t. $\cal G$, $\fs$, and $t$, then it is live w.r.t. $\cal G$, $\fs$, and $t$.
    \end{theorem}
    \begin{proof}
        Let $T\subseteq \Args$ be safe w.r.t. ${\cal G} = (G_1, \dots ,G_n)$, $\fs$, and $t$.
        This means $\forall G \in \cal G$, $\forall \argx \in T$ it holds that $\fs_{G}(\argx) \geq t$, and thus also $\fs_{G_n}(\argx)\geq t$.
        The previous statement is the condition for liveness, as $G_n$ is the end point of the chain.
    \end{proof}
    In general, the converse of Theorem~\ref{thm:strong-implies-weak-safety} does not hold, as shown in Example~\ref{ex:safety-example}.
    However, under certain restrictions on the expansion sets and the choice of semantics, a partial converse can be achieved.
    In the subsequent theorem, we propose two of such cases.
    \begin{theorem}
        \label{thm:weak-strong-safety-coincide}
        Suppose one of the following cases holds:
        \begin{enumerate}
            \item
            The chain $\cal G$ has length $1$, or
            \item
            $\cal G$ is a weak expansion chain and $\fs$ is a modular semantics.
        \end{enumerate}
        Then safe and live arguments coincide w.r.t. $\cal G$, $\sigma$, and $t$.
    \end{theorem}
    \begin{proof}
        Assume $T$ is live w.r.t. $\cal G$, $\fs$, and $t$, where the chain $\cal G$ is denoted by $(G_1,\dots ,G_n)$.
        Now, we prove our claim case by case.
        \begin{description}
            \item[Case (i):]
            If $\cal G$ has length $1$, then ${\cal G}=(G_{1})$, and hence $G_1=G_n$.
            Since the start and end point of $\cal G$ coincide, by Definition~\ref{def:weak-safety}, $\fs_{G_1}(\argx) \geq t$, for every $\argx \in T$.
            Since $\cal G$ is a singleton, $\forall G \in \cal G$, $\fs_{G}(\argx)\geq t$, which is the condition for safety.
            \item[Case (ii):]
            By the strong directionality argument~\cite{kampik-potyka} (Principle 4, Lemma 6.1), weak expansions under any modular semantics $\fs$ do not alter the final strengths of arguments in $G$.
            Thus, $\forall G, G'\in \cal G$ and $\argx \in T$, $\fs_G(\argx) = \fs_{G'}(\argx)$.
            Accordingly, it also holds $\forall \argx \in T$ that $\fs_{G_n}(\argx) \geq t$ if and only if  $\forall G \in {\cal G}: \fs_{G}(\argx) \geq t$.
        \end{description}
    \end{proof}
%


    \subsubsection*{Oscillation}
    \label{sec:liveness}
    Safety and liveness concern the behavior of topic arguments at or above the threshold of credibility, i.e., whether the arguments remain consistently credible, or eventually attain credibility at the end.
    However, these notions do not capture cases where arguments remain relatively stable, only to lose credibility at a few specific instances during (or at the end of) the dialogue.
    For instance, an argument may remain consistently above the threshold of credibility throughout the expansion chain, but drop below the threshold at the end point.
    Such an argument is then not live and hence not safe. However, it is still worthwhile to investigate the threshold-crossing behavior of such arguments.

    To study these cases, we introduce the notion of \emph{oscillation}.
    This notion aims to answer whether or not an argument's credibility remains stable over the whole expansion chain.
    The idea is formally introduced through fluctuations across the threshold of credibility.
    Intuitively, fluctuations quantify the ``flip-flopping'' behavior exhibited by an argument.
    Formally, we say that $\argx \in T$ \emph{moves across a threshold} $t$ w.r.t. $\fs$ and QBAGs $G$, $G'$ if and only if either $\fs_G(\argx) < t \leq \fs_{G'}(\argx)$, or $\fs_{G'}(\argx) < t \leq \fs_G(\argx)$ holds.
    This gives us a key prerequisite to define $k$-fluctuations as follows.
    \begin{definition}[$k$-Fluctuations]
        \label{def:fluctuations}
        We say that $\argx\in T$ shows $k$-\emph{fluctuations} w.r.t. ${\cal G}=(G_1,\dots ,G_n)$, $\fs$, and $t$ iff
        there exist natural numbers $1< i_1< \dots <i_k < n$ such that $\argx$ moves across the threshold $t$ w.r.t. $\fs$ and every $G_{i_j}, G_{i_{j+1}} \in \cal G$.
        We say $\argx$ shows \emph{exactly} $k$-\emph{fluctuations} if and only if $\argx$ shows at least $k$-fluctuations and does not show $k+1$-fluctuations.
    \end{definition}
    Intuitively, fluctuations are determined relative to the threshold of credibility, which must be (at least) met for threshold attainment.
    If an argument's final strength is equal to the threshold and attains it at the next QBAG in the chain, no fluctuations are counted.
    Similarly, no fluctuations are counted if an argument's final strength is above the threshold of credibility and becomes equal to the threshold at the next QBAG in the chain.
    We apply $k$-fluctuations to topic sets, which gives rise to the notion of \emph{oscillation}.
    \begin{definition}[Oscillation]
        \label{def:liveness}
        We say that a topic set $T$ \emph{oscillates} w.r.t. $\cal G$, $\sigma$, and $t$ if and only if each $\argx \in T$ shows at least $1$-fluctuation w.r.t. $\cal G$, $\fs$, and $t$.
        We say that a topic set is (at least / exactly) $n$-\emph{oscillating} if each of its arguments shows (at least / exactly) $n$-fluctuations.
    \end{definition}
    Let us illustrate these definitions with the help of an example.
    \begin{example}
        \label{ex:liveness-example}
        In Figure~\ref{fig:intro}, the chain ${\cal G} = (G, G', G'')$ has the topic set $\{\arga, \argb\}$ as $2$-oscillating w.r.t. the threshold $0.2$, which means that $\{\arga\}$ and $\{\argb\}$ cross the threshold at least twice.
        On the other hand, $\{\argc\}$ shows no fluctuation w.r.t. $t = 0.2$; accordingly, it is safe and not oscillating.
    \end{example}
    Intuitively, oscillating topic sets have the \emph{potential} to be live and safe, but in practice, they may \emph{actually} be live but \emph{never} safe.
    \begin{theorem}
        \label{thm:safety-characterization-using-liveness}
        If a topic set $T$ is safe w.r.t. $\cal G$, $\fs$, and $t$ then it does not oscillate w.r.t. $\cal G$, $\fs$, and $t$.
    \end{theorem}
    \begin{proof}
        Suppose $T$ is safe, then by definition $\fs_G(\argx) \geq t$,  $\forall G \in {\cal G}, \forall \argx \in T$.
        Now, if $T$ oscillates, then each $\argx\in T$ must cross the threshold $t$ at least once. This is impossible, as it would contradict the safety of $T$.
    \end{proof}

    On the other hand, oscillating topic sets can be live, given an appropriate number of fluctuations.
    This can be observed in Example~\ref{ex:liveness-example} and Figure~\ref{fig:intro}, where $\{\arga, \argb\}$ is live and exactly 2-oscillating.
    
    In certain cases, we can guarantee that no fluctuations are observed in the topic set.
    These include cases of weak expansions and evolutions that only append downstream arguments.
    We formally state this below as a corollary; the proof is a straight-forward application of Theorem~\ref{thm:weak-strong-safety-coincide} and Theorem~\ref{thm:safety-characterization-using-liveness}.
    \begin{corollary}
        \label{thm:weak-expansion-live}
        Suppose that one of the following cases holds:
        \begin{enumerate}
            \item
            The chain $\cal G$ has length $1$, or
            \item
            $\cal G$ is a weak expansion chain and $\fs$ is a modular semantics.
        \end{enumerate}
        Then, for all $\argx \in \Args$, $\argx$ shows no fluctuations w.r.t. $\cal G,\sigma$, and $t \in \mathbb{I}$. $\hfill \square$
    \end{corollary}
%

    \subsubsection*{\bf Discrete Fairness}
    \label{subsec:Fairness}
    Conceptually, our notions of fairness evaluate whether the topic arguments have an equal shot to achieve safety or liveness.
    Accordingly, these notions reflect the concept of \emph{resource distribution} (of final strengths equal to or greater than the threshold) among arguments in the topic set.
    Intuitively, these notions condition the safety or liveness of all topic arguments on the existence
    of safety or liveness in one topic argument in terms of a necessary condition.
    We introduce three notions of discrete fairness, namely \emph{ideally fair}, \emph{live fair}, and \emph{cautiously fair}.
    The notion of ideally fair stipulates that the safety of one topic argument implies the safety of the whole topic set.
    Live fair is the counterpart for liveness of ideally fair, i.e., if one topic argument is live, then the whole topic set is live.
    Lastly, cautiously fair is a weakening of the ideally fair condition.
    If any topic argument is safe, then the whole topic set is at least live.

    \begin{definition}[Discrete Fairness]
        \label{def:strength-fairness}
        We say a topic set $T$, w.r.t. ${\cal G}$,$\fs$, and $t$, is
        \begin{enumerate}
            \item
            \textbf{Ideally Fair}
            if and only if some $\{\argx\}\subseteq T$ is safe implies $T$ is safe.
            \item
            \textbf{Live Fair}
            if and only if some $\{\argx\}\subseteq T$ is live implies $T$ is live.
            \item
            \textbf{Cautiously Fair}
            if and only if some $\{\argx\}\subseteq T$ is safe implies $T$ is live.
        \end{enumerate}
    \end{definition}
    In our running example of Figure~\ref{fig:intro}, sets of arguments satisfy different fairness notions.
    \begin{example}
        Refer again to Figure~\ref{fig:intro}, and consider $\{\arga,\argb,\argc\}$ to be the topic set.
        Example~\ref{ex:safety-example} highlights that $\{\argc\}$ is safe w.r.t. $(G, G', G'')$ and threshold $0.2$.
        The final strengths of all topic arguments final attain $0.2$ at the end point $G''$, thereby making the set $\{\arga, \argb, \argc\}$ both cautiously fair and live fair.
        However, ideal fairness is violated by $\{\arga, \argb, \argc\}$ as $\{\argc\}$ is safe, whereas $\{\arga\}$ and $\{\argb\}$ are live but not safe.
    \end{example}
    The following connections are present between the discrete notions of fairness.
    \begin{theorem}
        For any topic set $T$, the following holds w.r.t. ${\cal G}$, $\fs$, and $t$:
        \begin{enumerate}
            \item
            Live fairness implies cautious fairness;
            \item
            Ideal fairness implies cautious fairness;
            \item
            Ideal fairness implies live fairness in the presence of a safe argument.
        \end{enumerate}
    \end{theorem}
    \begin{proof}
        For the first case, assume $T$ is live fair w.r.t. $\cal G$, $\fs$, and $t$.
        If there is a safe $\{\argx\} \subseteq T$, then by Theorem~\ref{thm:strong-implies-weak-safety}, $\{\argx\}$ is live, and consequently $T$ is live due to live fairness.
        This satisfies the constraint for being cautiously fair.

        For the second case, suppose $T$ is ideally fair, and $\{\argx\}\subseteq T$ is safe.
        Now, by ideal fairness of $T$, we get that $T$ is safe.
        By Theorem~\ref{thm:strong-implies-weak-safety}, $T$ is also live.

        For the third case, suppose $T$ is ideally fair with a safe topic argument $\argx$, which implies $T$ is safe, and by Theorem~\ref{thm:strong-implies-weak-safety}, $T$ is live, and therefore live fairness is satisfied.
    \end{proof}
    Additionally, given a safe topic argument, ideal fairness prohibits the oscillation of topic arguments, thereby implying an ``uninteresting'' behavior of the dialogue with respect to the concerned arguments.
    \begin{theorem}
        W.r.t. $\cal G$, $\fs$, and $t$, if $T$ is ideally fair, then $T$ does \emph{not} oscillate whenever there is a safe $\{\argx\}\subseteq T$.
    \end{theorem}
    \begin{proof}
        Suppose $T$ is ideally fair w.r.t. ${\cal G}$,$\fs$, and $t$, with a safe $\{\argx\} \subseteq T$.
        This implies every $\argx \in T$ satisfies $\fs_G(\argx) \geq t$ for every $G \in \cal G$, thus showing no fluctuations; accordingly, $T$ does not oscillate.
    \end{proof}

    \subsection*{Gradual Fairness}
    \label{sec:fluctuation-fairness}
    Conceptually, our notions of safety and liveness analyze whether the topic arguments are credible over the course of a dialogue: safety stipulates that the topic arguments attain the specified threshold, whereas liveness ensures threshold attainment at the end of the dialogue.
    Hence, we may ask \emph{to what extent} all topic arguments have equal opportunities to attain credibility over the course of the dialogue, i.e., how balanced the distribution of threshold-attaining arguments is among the topic set across the QBAGs in a given chain.

    To that end, we introduce gradual notions of fairness, aimed to measure the fairness of threshold-attainment among arguments in the topic set.
    Specifically, we consider two notions: Gini-based fairness and Shannon-based fairness.
    The former is based on the Gini index~\cite{gini1912}, which traditionally measures ``inequalities'' in the distribution of wealth, income, or consumption.
    In our setting, we measure how closely threshold attainment within a topic set is aligned with an \emph{ideal scenario}: in this (typically counterfactual) scenario, the total number of threshold-attaining instances within the topic set is equally distributed among each topic argument.

    \begin{definition}[Fairness Line and Safety Curve]
        Let $\cal G$ be a chain, $T$ be a non-empty finite set of topic arguments, and $t$ be a threshold. For any argument $\argx \in T$, we define
        \begin{equation}
            S_{{\cal G}, \fs, t}(\argx) := |\{G\in {\cal G} \ :\ \fs_G(\argx)\geq t \}| \ ,
        \end{equation}
        which is the number of QBAGs in $\cal G$ for which $\argx$ \emph{attains} $t$. We drop the subscripts whenever the context is clear. Let $S_{\text{total}} = \sum _{\argy \in T} S(\argy)$ denote the sum of threshold attainments for all the arguments.

        The \emph{fairness line} $f(x)$ is the linear segment joining $(0, 0)$ to $( |T|, S_{\text{total}} )$, defined piecewise as:
        \begin{equation}
            \label{eq:fairness-line}
            f(x):=
            \begin{cases}
                0 & \text{if } x < 0 ,
                \\
                \displaystyle \frac{S_{\text{total}}}{|T|} \cdot x & \text{if } 0 \leq x < |T| ,
                \\
                S_{\text{total}} & \text{if } x \geq |T| \ .
            \end{cases}
        \end{equation}

        Let $n = |T|$, and let $\argz_{0}, \dots, \argz_{n-1}$ be the arguments in $T$ arranged in increasing order of the number of times that each argument attains the threshold, i.e., if $0 \leq j < i < n$, then $S(\argz_j) \leq S(\argz_i)$.
        The \emph{safety curve} is a piecewise linear spline defined as:

        \begin{equation}
            \label{eq:safety-curve}
            \zeta_{{\cal G},T,\fs}(x) :=
            \begin{cases}
                0 & \text{if } x < 0 , \\
                S (\argz_0) \cdot x & \text{if } 0 \leq x < 1 ,
                \\
                \left(\sum \limits_{0 \leq j < i} S (\argz_j)\right)
                + S (\argz_i) \cdot (x - i) & \makecell[l]{\text{if } i \leq x < i+1 \\ \text{ for } i \in \{1, \dots , n-1 \}, n > 1,}
                \\
                S_{\text{total}} & \text{if } x \geq n \ . \\
            \end{cases}
        \end{equation}

    \end{definition}

    Gini-based fairness measures fairness in terms of the area enclosed between the fairness line and the safety curve.
    Formally, it is introduced as follows.
    \begin{definition}[Gini-based Fairness]
        \label{def:gini-based-fairness}
        For the topic set $T$, define the un-normalized Gini-based fairness score using Equation~\ref{eq:safety-curve} w.r.t. ${\cal G}$, $\fs$, and $t$:
        \begin{equation}
            \label{eq:gini-unnormalized}
            {\cal F}^{u}_{Gini}({\cal G}, T, \fs, t)
            :=
            \int_{0}^{|T|}|\zeta_{{\cal G}, \fs, t}(x) - f(x)| \ dx.
        \end{equation}
        Thus, the (normalized) Gini-based fairness score is defined as
        \begin{equation}
            \label{eq:gini-normalized}
            {\cal F}_{Gini}({\cal G}, T, \fs, t)
            :=
            \left(\frac{2}{{1 + e^{-{\cal F}^{u}_{Gini}({\cal G}, T, \fs, t)}}} - 1\right).
        \end{equation}
    \end{definition}
    We illustrate the fairness notion with the help of the following example.
    \begin{example}
        \label{ex:fluctuation-fairness-example}
        Consider the set of topic arguments $T=\{\arga, \argb, \argc\}$ in Figure~\ref{fig:intro}.
        Here, $\arga$ and $\argb$ remain above the threshold $t = 0.2$ for the QBAGs $G, G''$, w.r.t. the chain $(G, G', G'')$.
        The argument $\argc$, with $\{\argc\}$ being safe, attains $t$ for the entire chain, i.e., in all three QBAGs, while $\arga$ and $\argb$ attain $t$ in only two QBAGs. Thus, $T$ can be ordered as $(\arga, \argb, \argc)$, corresponding to $\langle \arga, 2 \rangle \preceq \langle \argb, 2 \rangle \prec \langle \argc, 3 \rangle$ when arranged by the number of threshold-attaining instances over the whole chain.

        Figure~\ref{fig:fairness-line} plots the \emph{fairness line}
        (dotted line) using Equation~\ref{eq:fairness-line}.
        This line represents the situation where the seven threshold-attainment instances are equally distributed among the topic arguments in $T$.
        The \emph{fairness line} connects the points $(0,0)$ and $(3,7)$, as per Equation~\ref{eq:fairness-line}.
        The \emph{safety curve} (solid line) in Figure~\ref{fig:fairness-line} is plotted as per Equation~\ref{eq:safety-curve}, and the actual fluctuations shown by $T$.
        Intuitively, the fairness line shows the cumulative sum of fluctuations w.r.t. the topic arguments according to the aforementioned ordering.
        Fairness is now measured in terms of the area enclosed between the safety curve and the fairness line.
        In the example, the Gini-based fairness score is $\approx 0.46$. 
        \begin{figure}
            \centering  \includegraphics[width=0.5\linewidth]{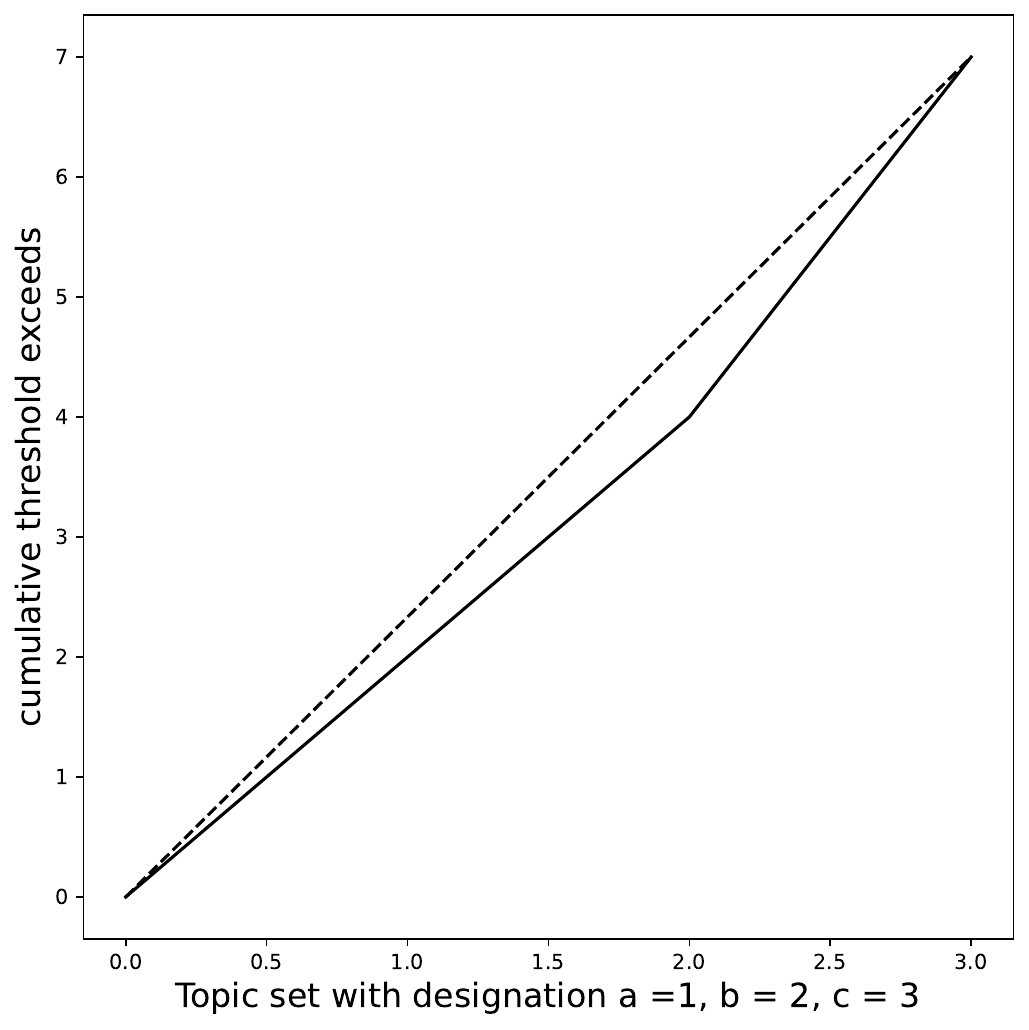}
            \caption{Plot of the fairness line versus $\zeta$.
            The solid line is due to the safety curve $\zeta$, and the fairness line is represented using dashes.
            The area enclosed between the two curves measures the inequality in fluctuation.}
            \label{fig:fairness-line}
        \end{figure}
    \end{example}

    As noted in Example~\ref{ex:fluctuation-fairness-example},
    the last equation ensures monotonicity of the Gini-based fairness notion; as the safety curve deviates further from the fairness line, ${\cal F}_{Gini}$ approaches 1.
    The case of perfect equality yields a score of $0$.
    \begin{theorem}
        \label{thm:well-definedness-of-gini-fairness}
        Fairness of fluctuations is a positive function having its range in the unit interval $[0,1)$.
        The value $0$ represents perfect equality, i.e., when the safety curve coincides with the fairness line.
    \end{theorem}
    \begin{proof}
        Notice that ${\cal F}^{u}_{Gini}({\cal G},T,\fs,t) \geq 0$ implies that $0\leq {\cal F}_{Gini}({\cal G},T,\fs,t)<1$.
        In the case where ${\cal F}^{u}_{Gini}({\cal G},T,\fs,t)=0$, i.e., when the fairness line and $\zeta$ coincide,
        $
        \left(\frac{2}{{1 + e^{-{\cal F}^{u}_{Gini}({\cal G}, T, \fs, t)}}}\right) = 1
        $
        and by definition, ${\cal F}_{Gini}({\cal G}, T, \fs, t) = 0$.
    \end{proof}
    \vspace{2 em}

    Our second gradual fairness notion is based on Shannon entropy~\cite{shannon_entropy}.

    \begin{definition}[Threshold-attainment Probability]
        The \emph{threshold-attainment probability} for the set of topic arguments $T$ w.r.t. ${\cal G}$, $\fs$, and a threshold $t\in \mathbb{I}$ is defined as
        \[
            p_{{\cal G}, t, \fs}(\argx)
            :=
            \frac{S_{{\cal G}, t, \fs}(\argx)}{\sum_{\argy\in T}S_{{\cal G}, t, \fs}(\argy)} \ ,
        \]

        which is the ratio between the total threshold-attaining QBAGs for $\argx$ in $\cal G$ and the total number of threshold-attaining QBAGs for all the arguments in the entire set of topic arguments $T$.
    \end{definition}

    We drop the subscript whenever it is clear from context.
    Notice that $p$ is \emph{undefined} in the absence of threshold-attainment behavior, i.e., if no argument in $T$ attains the threshold.

    Given an arbitrary threshold-attainment behavior, we want to measure how \emph{surprising} it is.

    \begin{definition}[Information Function]
        Formally, we define the \emph{information function} $I_{p}$ in terms of a topic argument $\argx$ with probability $p(\argx)$ as
        \[
            I_{p}(\argx)
            :=
            \log_b\left(\frac{1}{p(\argx)}\right) = -\log_b(p(\argx)) = \frac{- \log (p(\argx))}{\log b},
        \]
        where $b = |T|$.
    \end{definition}

    Note that $I_p(\argx)$ is either non-negative or \emph{undefined}.
    $I_p$ is undefined whenever $p_{{\cal G},\fs, t}(\argx)=0$ or $\log b=0$;
    the former corresponds to the cases where $\argx$ is consistently below the credibility threshold throughout the dialogue.
    When accounting for Shannon fairness, we leave out such cases, i.e, Shannon based fairness ignores consistently marginal arguments.
    On the other hand, the case where $\log b = 0$ occurs is when the topic set is a singleton, and in this case, we trivially assign the value $1$.

    Formally, Shannon-based fairness is defined as the expectation of the surprise factor.
    \begin{definition}[Shannon-based Fairness]
        \label{def:shannon-based-fairness}
        For the topic arguments $T$ w.r.t. the chain $\cal G$, semantics $\fs$ and threshold $t \in \mathbb{I}$, Shannon-based fluctuation fairness is defined as:
        \begin{equation}
            \label{eq:shannon-fairness-condensed}
            \mathcal{F}_{Shannon}({\cal G}, T, \fs, t) :=
            \begin{cases}
                \sum \limits _{\substack{\argx\in T : p(\argx)> 0}} p(\argx) \cdot I_{p} (\argx) & \makecell[l]{\text{if $p$ defined}, b > 1,}
                \\
                1 & \text{otherwise}.
            \end{cases}
        \end{equation}
        which is expanded to
        \begin{equation}
            \label{eq:shannon-fairness-expanded}
            \mathcal{F}_{Shannon}({\cal G}, T, \fs, t) :=
            \begin{cases}
                \displaystyle \frac{-1}{\log b} \cdot \sum \limits _{\substack{\argx\in T: p(\argx) > 0}}p(\argx) \cdot \log (p(\argx)) & \makecell[l]{\text{if $p$ defined}, b > 1,}
                \\
                1 & \text{otherwise}.
            \end{cases}
        \end{equation}
        where $b = |T|$.
    \end{definition}
    \begin{example}
        Let us determine the Shannon-based fairness for the topic set $T = \{\arga, \argb, \argc\}$ and threshold $t = 0.2$ in Figure~\ref{fig:intro}.
        We first compute the values of $S(\argx)$ for each $\argx \in T$, which are $\{ \langle \arga , 2 \rangle, \langle \argb , 2 \rangle, \langle \argc , 3 \rangle \}$, and the sum of these values is $7$.
        The probability of fluctuation is given by $p(\arga) = \frac{2}{7}$, $p(\argb) = \frac{2}{7}$, and $p(\argc) = \frac{3}{7}$.

        Since $b = 3$,
        \[
            \mathcal{F}_{Shannon}({\cal G},T,\fs,t) = \displaystyle \frac{- 1}{\log 3} \left( \frac{2}{7} \cdot \log \frac{2}{7} + \frac{2}{7} \cdot \log \frac{2}{7} + \frac{3}{7} \cdot \log \frac{3}{7} \right) \approx 0.98 \ .
        \]
    \end{example}
    The next result provides a case where Shannon-based fairness is maximized.
    \begin{theorem}
        \label{thm:shannon-fairness-uniform-distribution}
        If $\forall \argx,\argy \in T, S (\argx) = S (\argy)$ w.r.t. $\cal G, \fs$ and $t \in \mathbb{I}$,
        then $\mathcal{F}_{Shannon}({\cal G}, T, \fs) = 1$.
    \end{theorem}
    \begin{proof}
        If $\forall \argx \in T, S (\argx) = 0$, ${\cal F}_{Shannon}({\cal G}, T, \fs, t) = 1$ by definition, because $p$ is undefined.

        Suppose that $\forall \argx,\argy \in T, S (\argx) = S (\argy)$ and that $\forall \argx \in T,S (\argx) > 0$. This immediately implies that $\forall \argx, \argy, p(\argx) = p(\argy)$.
        This, in turn, implies $p(\argx) = \frac{1}{m}$, where $m = |T|$, and consequently $b = |T| = m$.
        Notice that for each $\argx \in T$,
        \[
            I_{p}(\argx) = -\log_b(p(\argx)) =  -\log _{m} \left(\frac{1}{m}\right) = 1 \ .
        \]
        Thus, the result is given by the following chain of equalities:
        \[
                {\cal F}_{Shannon}({\cal G}, T, \fs, t)
            =
            \sum_{\argx \in T}p(\argx) \cdot I_{p}(\argx) = \sum_{\argx \in T} \frac{1}{m} = 1 \ .
        \]
    \end{proof}
    Theorem~\ref{thm:shannon-fairness-uniform-distribution} shows the interplay between binary and gradual notions of fairness.
    If a topic $T$ is ideally fair in the presence of some strongly safe $\{\argx\}\subseteq T$, then ${\cal F}_{Shannon}({\cal G}, T, \fs, t) = 1$.
    Similarly, given the previous scenario, ${\cal F}_{Gini}({\cal G}, T, \fs, t) = 0$.
    \begin{theorem}
        If $T$ is ideally fair w.r.t. ${\cal G}, \fs, t$, then ${\cal F}_{Gini}({\cal G}, T, \fs, t) = 0$ and ${\cal F}_{Shannon}({\cal G}, T, \fs, t) = 1$ whenever there exists $\{\argx\} \subseteq T$ such that $\{\argx\}$ is safe w.r.t. ${\cal G}, \fs, t$.
    \end{theorem}
    \begin{proof}
        Notice that if $T$ is ideally fair, then one of its arguments is safe, and all the arguments in $T$ are safe, and all of them attain the threshold in all the QBAGs. Hence, $\forall \argx, \argy\in T$, we have $S (\argx) = S (\argy)$, and this value is the number of QBAGs in ${\cal G}$.
        This means $p$ is a uniform distribution, and by Theorem~\ref{thm:shannon-fairness-uniform-distribution}, ${\cal F}_{Shannon}({\cal G}, T, \fs, t) = 1$.

        The previous observation implies that $\zeta$ is a linear function satisfying $\zeta(0) = 0$ and $\zeta(|T|) = \sum_{\argx \in T} S (\argx)$.
        This, in turn, implies that $\zeta$ coincides with the fairness line, and by Theorem~\ref{thm:well-definedness-of-gini-fairness}, ${\cal F}_{Gini}({\cal G}, T, \fs, t) = 0$.
    \end{proof}

    \section{Discussion}
    \label{sec:discussion-and-conclusion}
    Our work contributes to the line of research on \emph{argumentation dynamics}, which, broadly speaking, studies argumentative inference, where the representation of the exchange of arguments undergoes changes.
    This naturally reflects the dynamic nature of dialectical reasoning.
    Argumentation dynamics are frequently studied for different variants of computational argumentation---see~\cite{doutre-argument} for a survey of works on abstract argumentation dynamics.

    For the specific case of QBAGs, dynamics are often applied in order to facilitate some notion of \emph{explainability}, notably by assessing the impact (a change to) one argument has on the final strength of another one~\cite{DBLP:conf/ecai/0007PT23,DBLP:journals/ijar/KampikPYCT24}.
    In these works, dynamics are thus essentially used as an analytical tool.
    Conversely, our work introduces analytical tools for studying dynamics that occur in sequences of QBAGs, which we call \emph{chains} or, intuitively, \emph{dialogues}.
    It may thus be most closely related to~\cite{KAMPIK-strength-change}, which studies which changes---also called \emph{explanations}---made to a QBAG affect the relative strength of two topic arguments.
    The current work could be extended accordingly to define explanations
    indicating why SOLF notions are satisfied or violated.


    \section{Future Work}
    \label{sec:future-work}
    Our study of SOLF notions for quantitative (bipolar) argumentation dialogues suggests the following notable directions for future work. 
    \\
    
    \noindent\textbf{Providing SOLF guarantees depending on argumentation semantics.}
    Our work analyzes only simple cases where some of our SOLF notions can be guaranteed, notably cases based on weak expansions, where changes do not affect the topic arguments.
    Indeed, we view the SOLF notions as analysis tools that can be applied to \emph{specific instances} of QBAG chains.
    Guaranteeing safety, oscillation, liveness, and fairness
    more generally can be challenging.
    Figure~\ref{fig:dfquad-inference-change} illustrates that even relatively simple scenarios can make it difficult to generally show that SOLF notions are satisfied or violated.

    Consider the QBAG in Figure~\ref{fig:graph-monotonic-effects}, with topic arguments $\arga$ and $\argb$, and assume that the strength of $\argf$ is variable. Consider a scenario where we have a chain of QBAGs that are all based on the one in the figure. Assume that the strength of $\argf$ increases, either due to an initial strength change or the influence of additional supporters, e.g., from $0.1$ to $0.5$ and finally to $0.9$.
    The effects of these changes on the final strengths of $\arga$ and $\argb$ are shown in Figure~\ref{fig:plot-non-monotonic-effects}.
    Notice that, intuitively, the effect of strength changes to $\argf$ on $\arga$ and $\argb$ is \emph{non-monotonic}: it ``switches'' from positive to negative for $\arga$ and inversely for $\argb$---this phenomenon is, e.g., discussed in~\cite{DBLP:journals/ijar/KampikPYCT24}.

    This makes it difficult to assess whether a set of topic arguments, even when consisting of a single argument, is safe or live, and analogously, whether a set of topic arguments is fair. In our example, $\{\arga\}$ is safe and live w.r.t. the thresholds $0.1$ and $0.175$, respectively. However, $\{\argb\}$ fails to satisfy both safety and liveness w.r.t. the above-mentioned thresholds.
    Both $\arga$ and $\argb$ oscillate w.r.t. the threshold $0.175$, but only $\argb$ oscillates w.r.t. the other threshold of $0.1$.
    $\{\argb\}$'s liveness failure prevents $\{\arga, \argb\}$ from satisfying the binary notions of fairness w.r.t. the thresholds $0.1$ and $0.175$. \\
    \begin{figure}
        \centering
        \subfloat[QBAG with variable initial stength of $\argf$.]{
            \centering

            \begin{tikzpicture}[node distance=1.5cm]
                \node[unanode]    (a)    at(0,0)  {\argnode{\arga}{0.2}{\phantom{0}\textbf{?}\phantom{0}}};
                \node[unanode]    (b)    at(4,0)  {\argnode{\argb}{0.0}{\phantom{0}\textbf{?}\phantom{0}}};
                \node[unanode]    (e)    at(0,2)  {\argnode{\arge}{0.0}{\phantom{0}\textbf{?}\phantom{0}}};
                \node[unanode]  (c)    at(2,2)  {\argnode{\argc}{0.0}{\phantom{0}\textbf{?}\phantom{0}}};
                \node[unanode]    (d)    at(4,2)  {\argnode{\argd}{0.0}{\phantom{0}\textbf{?}\phantom{0}}};
                \node[unanode]    (f)    at(2,4)  {\argnode{\argf}{?}{\phantom{0}\textbf{?}\phantom{0}}};

                \path [->, line width=0.2mm]  (e) edge node[left] {+} (a);
                \path [->, line width=0.2mm]  (c) edge node[left] {-} (a);
                \path [->, line width=0.2mm]  (d) edge node[below] {-} (a);
                \path [->, line width=0.2mm]  (e) edge node[above] {-} (b);
                \path [->, line width=0.2mm]  (c) edge node[right] {+} (b);
                \path [->, line width=0.2mm]  (d) edge node[right] {+} (b);
                \path [->, line width=0.2mm]  (f) edge node[left] {+} (e);
                \path [->, line width=0.2mm]  (f) edge node[left] {+} (c);
                \path [->, line width=0.2mm]  (f) edge node[left] {+} (d);
            \end{tikzpicture}
            \label{fig:graph-monotonic-effects}
        }
        \hfill
        \centering
        \subfloat[$\fs(\arga)$ (solid) and $\fs(\argb)$ (dashed), given $\tau(\argf)$.]{\includegraphics[width=0.5\linewidth]{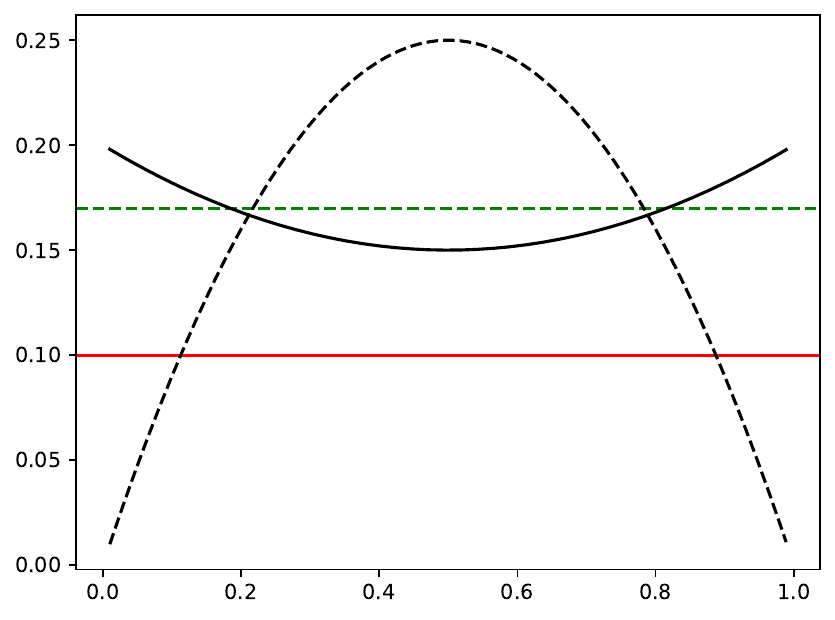}
        \label{fig:plot-non-monotonic-effects}
        }
        \caption{SOLF properties are difficult to guarantee, even in seemingly simple scenarios where we change the initial strength of a single argument that reaches two topic arguments. The horizontal lines in Figure~\ref{fig:plot-non-monotonic-effects}. indicate different threshold values (of $0.1$ and $0.175$, respectively).}
        \label{fig:dfquad-inference-change}
    \end{figure}

    \noindent \textbf{Defining SOLF notions for partially ordered final strengths.}
    In this work, we have made the crucial assumption that our assessments of safety, oscillation, liveness, and fairness pertain to topic arguments that are \emph{comparable}.
    However, this may not always be the case in QBAGs, as arguments may have \emph{undefined} final strengths, practically, because a modular semantics may not converge for cyclic graphs.
    Pragmatically, one may claim that undefined final strengths give rise to an edge case whose consideration introduces overhead and thus confusion.
    However, in some cases, dealing with arguments of incomparable strengths can be useful. This is useful, for instance, when moving from quantitative to labelling-based abstract argumentation, where labels may form a lattice of argument strengths (cf.~\cite{wu2010labelling}), as sketched in~\cite{KAMPIK-strength-change}. 
    \newpage
    
    \noindent\textbf{Introducing SOLF notions to other argumentation variants.}
    The aforementioned bridge to abstract argumentation, via labelling-based approaches that give rise to notions akin to final strengths, can enable the introduction of our SOLF notions to other variants of computational argumentation.
    In the context of abstract argumentation, our notions could then be integrated with a wealth of established results, e.g., regarding \emph{stability}, \emph{robustness}, and other invariant notions~\cite{Rienstra.al-2020-Robustness,Bistarelli.Santini.Taticchi-2018-InvariantOperators,Baumann-2012-NomalStrongExpansion,Odekerken.Borg.Bex-2022-StabilityAndRelevance}.
    Conversely, one could extend the scope of these existing notions to a quantitative setting.

    \subsubsection*{Acknowledgements}
    This work was partially supported by the Wallenberg AI, Autonomous Systems and Software Program (WASP) funded by the Knut and Alice Wallenberg Foundation.
    \newpage
    
    \bibliographystyle{src/main/latex/main/LNGAI}
    \bibliography{src/main/latex/main/main.bib}

\end{document}